%% file: 00_main.tex
\documentclass[runningheads]{llncs}
\let\llncssubparagraph\subparagraph
\let\subparagraph\paragraph
\usepackage[compact]{titlesec}
\let\subparagraph\llncssubparagraph

\usepackage{wrapfig}

\usepackage[T1]{fontenc}
\usepackage{graphicx}
\usepackage{amsmath}
\usepackage{amssymb}

\usepackage{hyperref} 
\usepackage{booktabs} 
\newcommand{\ra}[1]{\renewcommand{\arraystretch}{#1}}
\usepackage{csquotes}
\usepackage{pifont}    
\usepackage{xcolor}    
\usepackage{bm}

\newcommand{\greentick}{\textcolor{black}{\ding{51}}} 
\newcommand{\reddash}{\textcolor{gray}{\textemdash}}

\usepackage{listings}

\newcommand{\highlightgray}[1]{\colorbox{gray!20}{#1}}

\usepackage{subcaption}

\newcommand{\WD}{\mathtt{WD}}
\newcommand{\AWD}{\mathtt{AWD}}

\title{Augmented Weak Distance for Fast and Accurate Bounds Checking}
\titlerunning{Augmented Weak Distance}

\author{Zhoulai Fu\inst{1,3} \ \ \ \  Freek Verbeek\inst{2,3} \ \ \ \ Binoy Ravindran\inst{3}}
\authorrunning{Z. Fu, F. Verbeek, and B. Ravindran}

\institute{
State University of New York, Korea \and
Open Universiteit of the Netherlands, Netherlands \and
Virginia Tech, USA
}

\usepackage[linesnumbered,ruled,vlined]{algorithm2e}

\usepackage{tcolorbox} 

\usepackage{enumitem}
\setlist{topsep=2pt, itemsep=2pt, partopsep=1pt}

\SetAlgoNlRelativeSize{-1} 

\titlespacing*{\section}{0pt}{2ex}{1ex}
\titlespacing*{\subsection}{0pt}{1.5ex}{1ex}
\titlespacing*{\subsubsection}{0pt}{1ex}{0.5ex}

\begin{document}

\maketitle


 \input{10_abstract}

\input{01_intro}

\input{02_eg}


 \input{03_method}



\input{05_results}

\input{06_discussion}

 \section{Related Work}

\input{07_relwork}

 \label{sect:relework}

\section{Conclusion}
\input{08_conc}

\bibliographystyle{splncs04}
\bibliography{references}

\end{document}

%% file: 10_abstract.tex
\begin{abstract}

This work advances floating-point program verification by introducing Augmented Weak-Distance (AWD), a principled extension of the Weak-Distance (WD) framework. WD is a recent approach that reformulates program analysis as a numerical minimization problem, providing correctness guarantees through non-negativity and zero-target correspondence. It consistently outperforms traditional floating-point analysis, often achieving speedups of several orders of magnitude. However, WD suffers from ill-conditioned optimization landscapes and branching discontinuities, which significantly hinder its practical effectiveness. AWD overcomes these limitations with two key contributions. First, it enforces the \emph{Monotonic Convergence Condition} (MCC), ensuring a strictly decreasing objective function and mitigating misleading optimization stalls. Second, it extends WD with a per-path analysis scheme, preserving the correctness guarantees of weak-distance theory while integrating execution paths into the optimization process. These enhancements construct a well-conditioned optimization landscape, enabling AWD to handle floating-point programs effectively, even in the presence of loops and external functions. We evaluate AWD on SV-COMP 2024, a widely used benchmark for software verification.Across 40 benchmarks initially selected for bounds checking, AWD achieves 100\% accuracy, matching the state-of-the-art bounded model checker CBMC, one of the most widely used verification tools, while running 170X faster on average. In contrast, the static analysis tool Astrée, despite being fast, solves only 17.5\% of the benchmarks. These results establish AWD as a highly efficient alternative to CBMC for bounds checking, delivering precise floating-point verification without compromising correctness.

\end{abstract}

%% file: 01_intro.tex
\footnotesize   
\begin{displayquote} 
\textit{There is absolutely no doubt that every effect in the universe can be explained satisfactorily from final causes, by the aid of the method of maxima and minima.} 
    
    \begin{flushright}
    - Leonhard Euler, 1744~\cite{euler1952methodus}.
    \end{flushright}
\end{displayquote}
\normalsize

\section{Introduction}\label{sect:intro}

Handling floating-point semantics remains a fundamental challenge in programming languages. Floating-point programs are difficult to reason about (e.g., $0.1 + 0.2=0.30000000000000004$ in modern hardware) and tricky to implement correctly. A straightforward computation like $\sqrt{x + 1} - 1$ in C can produce a 12\% relative error for small $x$ due to floating-point cancellation~\cite{DBLP:conf/oopsla/FuBS15}.

The \emph{weak-distance approach} transforms floating-point program analysis into a mathematical optimization problem~\cite{DBLP:conf/pldi/FuS19}. It encodes program semantics using a generalized metric, called \emph{weak distance}, which precisely captures a given analysis objective. Minimizing this function is theoretically guaranteed to solve the problem, providing a systematic framework for numerical program analysis. Unlike heuristic methods that minimize ad-hoc objective functions without guarantees, weak distance ensures correctness by guaranteeing that the minimizer corresponds to a valid solution.

To illustrate the weak-distance approach, we present a simple example. 

\subsubsection{The Weak-Distance Approach}
Let \(\langle \mathit{Prog}, I \rangle\) be a floating-point analysis problem, where \(I \subset \text{dom}(\mathit{Prog})\) represents the set of inputs that trigger a specific program behavior. A function \(W: \text{dom}(\mathit{Prog}) \to \mathbb{F}\) is a \emph{weak distance} for \(\langle \mathit{Prog}, I \rangle\) if:

\begin{enumerate}
    \item[] \textbf{Non-Negativity} \(W(x) \geq 0\) for all \(x\).
    \item[] \textbf{Zero-Target Correspondence} \(W(x) = 0\) if and only if \(x \in I\).
\end{enumerate}

These properties reduce the analysis problem to minimizing \(W\). If the minimum is \(0\), the corresponding input belongs to \(I\). Otherwise, \(I\) is empty.

Consider the following example:

\begin{lstlisting}[language=C, caption={Example: Handling Loops Without Path Explosion}, label=list:check_sum]
void check_sum(double x) {
    int integral = (int)x;          
    double decimal = x - integral; // fractional part
    
    double sum = 0;
    for (int i = 1; i <= integral; i++) {
        sum += i;
    }
    if (sum + decimal == 11) printf("Unexpected");
}
\end{lstlisting}

This function checks whether an input \texttt{x} satisfies \texttt{sum + decimal == 11}, where \texttt{sum} is the sum of all integers from 1 to the integral part of \texttt{x}, and \texttt{decimal} is the fractional part.

At first glance, \texttt{"unexpected"} seems impossible to trigger: \(\sum_{i=1}^4 = 10\), \(\sum_{i=1}^5 = 15\), and adding any fractional part  cannot yield 11. However, an input \(x\) close to 5, such as \(4.s\) where \(0.s\) approaches 1, can satisfy the condition.

Traditional symbolic execution can find this input but often suffers from \emph{path explosion} and the complexity of floating-point constraint solving. The weak-distance approach offers an efficient alternative by embedding a numerical deviation metric into the program. Specifically, it introduces a global variable \texttt{d} that captures the absolute difference between the left-hand and right-hand sides of each branch condition. In this case:
\begin{lstlisting}
d = fabs(sum + decimal - 11);
\end{lstlisting}
 This transformation defines a function \(\mathit{W}\) that satisfies E1 and E2, enabling us to solve the problem via numerical minimization.

Despite discontinuities in the weak-distance landscape, modern optimization techniques efficiently identify this minimum. By using \emph{Basinhopping}~\cite{iwamatsu2004basin}, which combines MCMC-inspired stochastic sampling with local minimization, we uncover an unexpected solution: 
\begin{align}
    \texttt{x = 4.999999999999999}, 
\end{align}
which is distinct from 5 (no floating-point representation issue). This value produces \texttt{sum = 10} and \texttt{decimal = 0.999999999999999}, satisfying \texttt{sum + decimal == 11} and triggering \texttt{"unexpected"}.

\begin{figure}[ht!]
    \centering
    \begin{subfigure}[b]{0.47\textwidth}
        \begin{lstlisting}
(*@\highlightgray{double d;}@*)
double W(double x){
  ... for-loop 
  (*@\highlightgray{d = fabs(sum + decimal - 11);}@*)
  if (sum + decimal == 11) {   
     printf("Unexpected");
}
        \end{lstlisting}
        \caption{}
    \end{subfigure}
    \hfill
    \begin{subfigure}[b]{0.47\textwidth}
        \includegraphics[width=\textwidth]{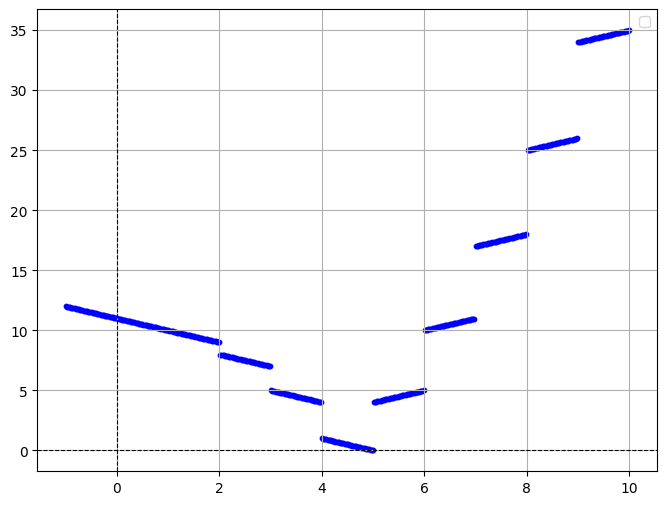}
        \caption{}
    \end{subfigure}
    \caption{(a) A weak-distance for \texttt{check\_sum}, and (b) its visualization}
    \label{fig:weak_distance_check_access}
\end{figure}

This weak-distance approach avoids direct reasoning about floating-point semantics, allowing it to handle loops and external functions seamlessly. It has been successfully applied to satisfiability solving~\cite{DBLP:conf/cav/FuS16}, automated testing~\cite{DBLP:conf/pldi/FuS17}, and boundary value analysis~\cite{DBLP:conf/pldi/FuS19}, often achieving an order-of-magnitude performance improvement over existing techniques. 

\subsubsection{Augmented Weak-Distance}
Despite its advantages, weak distance suffers from two key limitations: (1) its constraints of non-negativity and zero-target correspondence are \emph{too weak}, leading to \emph{flat landscapes} (vanishing gradients) or \emph{misleading objectives}; (2) it struggles with \emph{discontinuities from branching}, which hinder optimization.

We address these issues with \emph{augmented weak-distance}, contributing both theoretically and practically. Theoretically, we introduce an additional constraint, \(\text{MCC}\), formalized via \emph{path-input affinity}, which quantifies input closeness to an execution path by considering: (i) {execution depth}, tracking how far execution follows the expected path, and (ii) {branch satisfaction}, measuring how well branch conditions align with expected behavior. Embedding \emph{path-input affinity} ensures a structured optimization landscape, overcoming flat or discontinuous regions.
Practically, we develop a \emph{systematic algorithm} for constructing augmented weak-distance functions. This algorithm dynamically traces execution paths and computes path affinity at runtime with minimal overhead, ensuring efficient floating-point analysis. As an alternative to symbolic execution, it avoids SMT solving and loop unwinding, enabling robust floating-point verification even in the presence of loops and external calls.

\subsubsection{Empirical Results}
We evaluated our method on the \emph{SV-COMP 2024 benchmark suite}~\cite{svcomp2024}. Across 40 benchmarks with known ground truths:
\begin{itemize}
    \item Our approach \textbf{achieved 100\% accuracy}, matching the state-of-the-art bounded model checker CBMC~\cite{DBLP:journals/corr/abs-2302-02384}.
    \item Our method was \textbf{170 times faster} than CBMC.
    \item The static analysis tool Astrée~\cite{DBLP:conf/asian/CousotCFMMMR06} was \textbf{faster} but solved only \textbf{17.5\%} of the benchmarks.
\end{itemize}

%% file: 02_eg.tex
\section{Weak Distance: Pitfalls and How We Address Them}\label{sect:eg}

As introduced in Sect.~\ref{sect:intro}, the weak-distance approach embeds a global variable \texttt{w = ...} before each branch, capturing deviations from satisfying the branch condition. The weak distance function satisfies \emph{non-negativity} and \emph{zero-target correspondence}, theoretically ensuring that minimizing \texttt{d} leads to a solution.  

However, theoretical correctness does not guarantee practical feasibility. Weak distance may mislead optimization due to abrupt changes in the landscape or discontinuities in execution paths, leading to premature convergence or missing valid solutions. We identify two key pitfalls:  
(1) \emph{Non-Monotonic Descent Disrupts Optimization}—Abrupt jumps in the landscape mislead the optimizer, causing premature convergence.  
(2) \emph{Ignoring Per-Path Behavior Causes Incorrect Solutions}—Weak distance assumes a globally smooth landscape, but execution paths may behave discontinuously.

The following sections illustrate these pitfalls and outline key principles for addressing them.

\subsection{Example 1: Non-Monotonic Descent Disrupts Optimization} \label{sect:mcc}

Consider the following example:

\begin{lstlisting}[language=C, caption={Sample function with nested if statements.}]
void check_date(int day, int month) {
    if (day == 20) {
        if (month == 10) 
            printf("reached");
    }
}
\end{lstlisting}

\subsubsection{How Weak Distance is Normally Applied} 

To construct a weak-distance function, we introduce \texttt{w} as a global variable to capture how far the input deviates from satisfying key branch conditions. At each conditional check, \texttt{w} encodes the squared difference between the input variable and the expected value. This ensures non-negativity and zero-target correspondence but does not guarantee an informative gradient for optimization. Fig.~\ref{fig:w3_vs_aw} illustrates this weak distance construction.

\begin{figure}[ht!]
    \centering
    \begin{minipage}{0.45\textwidth}
        \begin{lstlisting}[language=C]
double WD(double d, double m){
  int day = floor(d);
  int month = floor(m);
  (*@\highlightgray{d = (day - 20) ** 2;}@*) 
  if (day == 20){ 
    (*@\highlightgray{d = (month - 10) ** 2;}@*) 
    if (month == 10) 
      printf("reached");
}
        \end{lstlisting}
    \end{minipage}
    \hfill
    \begin{minipage}{0.45\textwidth}
        \begin{lstlisting}[language=C]
double AWD(double d, double m){
  int day = floor(d);
  int month = floor(m);
  (*@\highlightgray{d = (day - 20) ** 2 + 150;}@*) 
  if (day == 20){ 
    (*@\highlightgray{d = (month - 10) ** 2;}@*) 
    if (month == 10) 
      printf("reached");
}
        \end{lstlisting}
    \end{minipage}
    \caption{C Weak-Distance (L) and Augmented Weak-Distance (R) for \texttt{check\_sum}.}
    \label{fig:w3_vs_aw}
\end{figure}

\subsubsection{Why Weak Distance Fails}
Although \(\WD\) satisfies non-negativity and zero-target correspondence, it fails to guide optimization toward the correct target \( (20,10) \). The failure can be observed in how the minimization process unfolds:  
(1) The optimizer minimizes \(d\), reducing the squared distance from a starting point $(1,1)$.  
(2) It successfully moves along the \(x\)-axis, reaching \((19,1)\) where \(d = 1\).  
(3) However, when \(x = 20\), \(d\) jumps to 81, increasing sharply.  
(4) The optimizer misinterprets this jump as an incorrect direction, causing it to halt prematurely at $(19,1)$.

Thus, the weak distance fails because closeness does not consistently correspond to a lower objective function value. A strictly decreasing objective function as the input approaches the solution is needed. 



\subsubsection{Fixing the Issue: Monotonic Convergence Condition (MCC)}

To resolve this, we introduce the \emph{Monotonic Convergence Condition (MCC)}, ensuring a strictly decreasing objective function as the input approaches the solution: 

\begin{quote}
    \emph{MCC:} The objective function must decrease monotonically as the solution is approached.
\end{quote}
\vspace*{-3ex}

MCC accounts for two key aspects of closeness: (1) \emph{Execution depth}—the objective should decrease as the realized execution path nears the target in the syntax tree, and (2) \emph{Branch satisfaction}—the distance between the left- and right-hand sides of a branch, as already considered in weak distance. We formalize this in Sect.~\ref{sect:method} as \emph{path-input affinity}.

Fig.~\ref{fig:w3_vs_aw} modifies \(\WD\) to construct an \emph{Augmented Weak-Distance} function \(\AWD\) by introducing a large offset (e.g., 150) before the first branch:

\begin{equation}
    \text{d} = ( \text{day} - 20 )^2 + 150.
\end{equation}

This adjustment ensures that before satisfying the first branch (\(\text{day} = 20\)), the function remains guided by \((\text{day} - 20 )^2\), while the added 150 has no effect. Once the first branch is met, the function transitions to the second branch, where it is now governed by \((\text{month} - 10)^2\), a strictly smaller value than 150. This naturally guides the optimization deeper into the program, leading to \( (\text{day} = 20, \text{month} = 10) \).

Fig.~\ref{fig:aw3_w3} illustrates how \(\AWD\) corrects \(\WD\)'s misleading behavior by reshaping the optimization landscape. Unlike \(\WD\), which introduces structural discontinuities that hinder convergence, \(\AWD\) ensures a smoother descent, enabling the optimizer to reach the desired solution efficiently.

\begin{figure}[t]
    \centering
    \begin{subfigure}{0.48\textwidth}
        \centering
Have a few        \includegraphics[width=\linewidth]{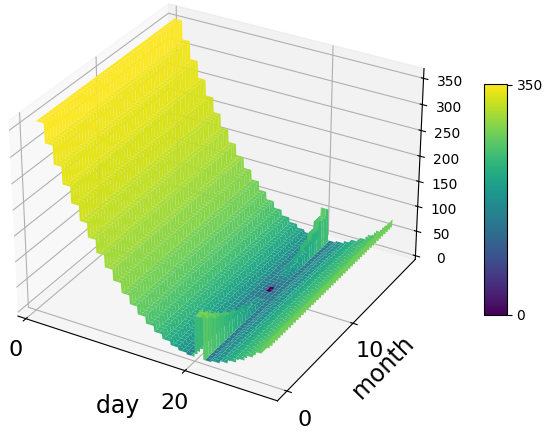}
    \end{subfigure}
    \hfill
    \begin{subfigure}{0.48\textwidth}
        \centering
        \includegraphics[width=\linewidth]{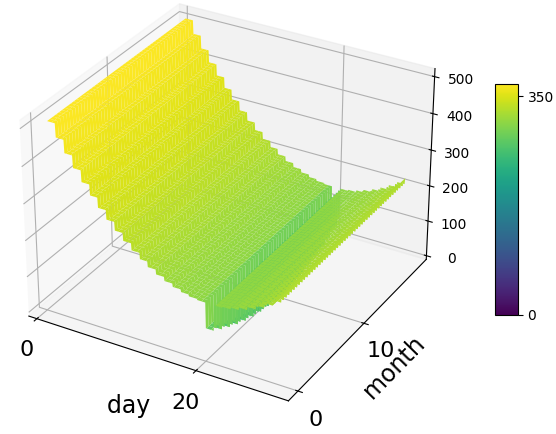}
    \end{subfigure}
    \caption{Weak distance (L) and Augmented Weak Distance (R) for \texttt{check\_sum}. }
    \label{fig:aw3_w3}
\end{figure}

\subsection{Example 2: Ignoring Per-Path Behavior Misses Solutions}\label{sect:per_path}

Numerical programs often use the \texttt{safe\allowbreak\_reciprocal} function:
\[
\text{safe\_reciprocal}(x) =
\begin{cases}
  1/x, & x \neq 0 \\
  0, & x = 0
\end{cases}
\]
to prevent division-by-zero, avoiding numerical instability and \texttt{NaN} propagation~\cite{bengio2012practical}. Consider the following program, which computes the cotangent of \( x \) using \texttt{safe\_reciprocal}, and checks whether the result is zero:
\begin{lstlisting}[language=C, caption={Program using \texttt{safe\_reciprocal} to compute \(\cot(x)\).},label={lst:cotangent}]
double cot(double x) {
   double y = safe_reciprocal(tan(x)); 
   if (y == 0) printf("reach 0"); 
   return y; }
\end{lstlisting}

\subsubsection{How Weak Distance is Normally Applied} 
The weak-distance method attempts to locate inputs that trigger \texttt{"reach 0"}. A standard weak-distance function embeds \texttt{d = (y - 0) ** 2;} before the conditional statement, creating a function that maps the input \(x\) to the global variable \texttt{d}.

\begin{lstlisting}
double WD(double x) {
    double y = safe_reciprocal(tan(x));
    (*@\highlightgray{d = (y - 0) ** 2;}@*) // Weak distance metric
    if (y == 0) printf("reach 0");
    return y;}
\end{lstlisting}

\subsubsection{Why Weak Distance Fails}
The weak-distance function creates a misleading optimization landscape, steering the optimizer away from \( x = 0 \). The failure arises because weak distance \emph{lacks smoothness across its input range}. Specifically, as \( x \) approaches \( \pi/2 \), \( y \) approaches 0, and the weak-distance function provides a smooth gradient, guiding the optimizer toward this solution. However, at \( x = 0 \), weak distance increases, incorrectly signaling the optimizer to move away, despite it being a valid solution.

\subsubsection{Fixing the Issue: Per-Path Optimization}
To overcome this limitation, we introduce a \emph{per-path optimization approach}, where each execution path is treated independently to avoid misleading discontinuities. Specifically, we separate execution paths: one solving for \( x = 0 \), where \(\tan(x) = 0\), and another solving for \( x = \pi/2 \), where \(\tan(x) \to \infty\). By treating each path separately, we \emph{restore smoothness within each optimization step}, ensuring that numerical minimization correctly identifies all solutions. A potential downside of this approach is the need for explicit path enumeration. However, weak distance remains advantageous in handling loops and external function calls, as it naturally constructs numerical objectives that generalize across execution paths.

\subsubsection{Takeaways}
These examples highlight a fundamental limitation of weak distance: its properties of \emph{non-negativity and zero-target correspondence alone are insufficient} when structural or numerical discontinuities exist in the optimization landscape. Weak distance fails because it treats all branches equally, ignoring both execution depth and path-dependent behavior. To overcome this, optimization must align with program structure and explicitly account for execution paths. The \emph{Monotonic Convergence Condition} (MCC) ensures a strictly decreasing objective function, preventing misleading optima, while the \emph{per-path optimization strategy} restores continuity and eliminates weak-distance misguidance. Together, these enhancements enable robust and accurate program analysis, even in the presence of branching, loops, and floating-point instability.

%% file: 03_method.tex
\section{Foundations of Augmented Weak Distance} \label{sect:method}

\subsubsection{General Notation} 

We denote the set of non-negative integers by \(\mathbb{Z}_{\geq 0}\), representing all integers greater than or equal to zero. The set of finite floating-point numbers in IEEE 754 (excluding NaN and infinity)~\cite{cowlishaw2008standard} is denoted by \(\mathbb{F}\). Given a function \(f\), its input domain is written as \(\mathit{dom}(f)\). The Cartesian product of two sets \(A\) and \(B\) is denoted as \(A \times B\). The notation \(\lambda x. f(x)\) defines an \emph{anonymous function}, representing a function that takes an input \(x\) and returns \(f(x)\). 

We use \(\mathit{Prog}\) to denote a program, \(\mathit{Br}\) for a branch, \(\pi\) for a path, and \(x\) for an input. A path is a sequence of branches, where each branch consists of a label and a Boolean outcome (true or false). . A path is \emph{partial} if it forms a prefix of a complete execution path.  Each input \(x\) is treated as an \(n\)-dimensional floating-point vector, meaning that \(\mathit{dom}(\mathit{Prog}) \subseteq \mathbb{F}^n\) for some integer \(n\). 

The notation \(\mathit{Prog}(x)\) represents the execution path taken by \(\mathit{Prog}\) when executed with input \(x\), which is a sequence of encountered branches.

\begin{definition}[Bounds Checking]
Bounds checking is a reachability problem that determines whether a numerical bound, specified by an arithmetic comparison (\(<, \leq, >, \geq, =\)) between floating-point numbers, can be reached. Given a target program \(\mathit{Prog}\) and a branch \(\mathit{Br}\), the goal is to find an input \(x\) such that the execution path \(\mathit{Prog}(x)\) passes through \(\mathit{Br}\).
\end{definition}

\subsection{Overall Solution}

We solve the bounds checking problem using a per-path  scheme to maximize precision. Specifically, for each execution path \(\pi\) passing through \(\mathit{Br}\), we seek an input that triggers \(\pi\). The per-path strategy mitigates the risk of missing solutions (Sect.~\ref{sect:per_path}), but a key challenge arises in handling loops. However, as illustrated in List.~\ref{list:check_sum}, weak distance  mitigates this issue by disregarding loops.  Therefore, we start our solution with the following  path synthesis  step.

\begin{center}
\begin{tcolorbox}[width=0.95\textwidth, colframe=gray!60, colback=gray!5, boxrule=0.5pt]
\paragraph{Step 1: Path Synthesis} 
We collect all execution paths terminating at \(\mathit{Br}\), excluding those containing loops or external function calls.
\end{tcolorbox}
\end{center}

Let \(\mathit{Paths}\) denote the set of synthesized paths from Step 1. To systematically evaluate each path, we construct an augmented weak distance function \(\AWD\), extending weak distance to a per-path formulation:
\[
\AWD: \mathit{Paths} \times \mathit{dom}(\mathit{Prog}) \to \mathbb{F}.
\]
The function \(\AWD(\pi, x)\) takes a path \(\pi \in \mathit{Paths}\) and an input \(x \in \mathit{dom}(\mathit{Prog})\), returning a floating-point value that quantifies how close \(x\) is to triggering \(\pi\). We impose the following conditions to ensure its correctness:

\begin{itemize}
    \item[] {\bf Non-negativity}: For all \(\pi \in \mathit{Paths}\) and \(x \in \mathit{dom}(\mathit{Prog})\), \(\AWD(\pi, x) \geq 0\).
    \item[] {\bf Zero-target correspondence}: \(\AWD(\pi, x) = 0\) if and only if \(x\) triggers \(\pi\). 
\end{itemize}

To address non-monotonic descent (Sect. \ref{sect:mcc}), we impose:

\begin{itemize}
    \item[] {\bf MCC}: For each path \(\pi \in \mathit{Paths}\), the function \(\lambda x. \AWD(\pi, x)\) must decrease as \(x\) moves closer to the target branch.
\end{itemize}

We will formalize  "closeness" in the next section. Below is Step 2.

\begin{center}
\begin{tcolorbox}[width=0.95\textwidth, colframe=gray!60, colback=gray!5, boxrule=0.5pt]
\paragraph{Step 2: AWD construction} Construct an $\AWD$ function satisfying non-negativity, zero-target correspondence, and MCC conditions.
\end{tcolorbox}
\end{center}

Once the \(\AWD\) function is constructed, we proceed to minimize it to solve the  bounds checking problem. Non-negativity and zero-target correspondence ensure that if \(\min_{x} \AWD(\pi, x) = 0\) for any \(\pi\), then the corresponding minimizer (\(\arg\min_{x} \AWD(\pi, x)\)) represents an input that triggers \(\pi\). The MCC condition guarantees that this optimization process is feasible. Thus, we conclude with:

\begin{center}
\begin{tcolorbox}[width=0.95\textwidth, colframe=gray!60, colback=gray!5, boxrule=0.5pt]
\paragraph{Step 3:} Minimize \(\lambda x . \AWD(\pi, x)\) for each \(\pi\). Report inputs $x$ where \(\AWD(\pi, x) = 0\).
\end{tcolorbox}
\end{center}

\subsection{Path-Input Affinity: Formalizing MCC}

Let \(\pi\) be a partial execution path terminating at \(\mathit{Br}\), and let \(\mathit{Prog}\{x\}\) denote the path taken by executing \(\mathit{Prog}\) with input \(x\). To formally define the notion of \emph{closeness} between \(\pi\) and \(x\) required for the Monotonic Convergence Condition (MCC), we introduce a metric called \emph{path-input affinity}. As motivated in Sect.~\ref{sect:mcc}, this metric captures two key aspects: (1) how deeply the realized path aligns with the expected path, and (2) how close the branch condition at the fork (where \(\mathit{Prog}\{x\}\) and \(\pi\) diverge) is to being satisfied.

\subsubsection{Preliminary Definitions}

We first introduce the necessary building blocks to define path-input affinity rigorously.

The \emph{Floating-Point Cardinality} between two floating-point numbers \(a, b \in \mathbb{F}\) is defined as the number of distinct floating-point values \(x \in \mathbb{F}\) such that \(\min(a, b) \leq x < \max(a, b)\). We denote this count by \(\kappa(a, b)\). Clearly, \(0 \leq \kappa(a, b) \leq 2^{64}\) for double-precision floating-point numbers.

We define the \emph{lexicographic order} \(\prec\) over \(\mathbb{Z}_{\geq 0} \times \mathbb{Z}_{\geq 0}\) as:
\[
(x', y') \prec (x, y) \quad \iff \quad x' < x \quad \text{or} \quad (x' = x \text{ and } y' < y).
\]
For example, \((3, 27) \prec (10, 1)\) and \((3, 27) \prec (3, 42)\).

The \emph{fork branch}, or simply \emph{fork}, between two paths is the first branch where they diverge. If one path subsumes the other, the fork does not exist.

\begin{definition}[Path-Input Affinity]
Given a bounds checking problem \((\mathit{Prog}, \mathit{Br})\), let \(\mathit{Paths}\) denote the set of partial paths ending at \(\mathit{Br}\). The \emph{path-input affinity} quantifies the closeness between a path \(\pi\) and an input \(x\) as:
\[
\mathit{Aff}: \mathit{Paths} \times \text{dom}(\mathit{Prog}) \to \mathbb{Z}_{\geq 0}^2,
\]
where \(\mathit{Aff}(\pi, x) = (u, v)\) is defined as follows: \(u\) is the number of branches in \(\pi\) starting from the fork branch; if the fork does not exist, \(u = 0\); \(v\) is the number of floating-point values between the left-hand side (lhs) and right-hand side (rhs) of the fork branch; if no fork exists, \(v = 0\).

The precise computation of \(v\) is given by:
\[
v =
\begin{cases} 
    \kappa(\text{lhs}, \text{rhs}) & \text{if the fork branch is } \text{lhs} \leq \text{rhs}, \text{lhs} \geq \text{rhs}, \text{or } \text{lhs} = \text{rhs}, \\
    \kappa(\text{lhs}, \text{rhs}) + 1 & \text{if the fork branch is } \text{lhs} < \text{rhs}, \text{lhs} > \text{rhs}, \\
    1 & \text{if the fork branch is } \text{lhs} \neq \text{rhs}.
\end{cases}
\]
\end{definition}\label{def:path_affinity}

\subsubsection{Example}

Consider the following program:
\begin{lstlisting}[language=C]
void foo(double x) {
    if (x <= 3.0) error();
}
\end{lstlisting}
Let the target path be \(\pi = [0T]\), where \(0\) denotes the label for the \texttt{if} condition. Then:
\[
\mathit{Aff}(\pi, 3.2) = (1, \kappa(3.2, 3)) \quad \text{and} \quad \mathit{Aff}(\pi, 2.9) = (0, 0).
\]

\subsubsection{Formalizing MCC with Path-Input Affinity}

With the notion of path-input affinity established, we can now express the MCC condition formally. For any \(\pi \in \mathit{Paths}\) and any two inputs \(x, x' \in \mathit{dom}(\mathit{Prog})\),
\[
\mathit{Aff}(\pi, x') \prec \mathit{Aff}(\pi, x) \quad \Rightarrow \quad \AWD(\pi, x') < \AWD(\pi, x).
\]
This guarantees that as an input \(x\) becomes \emph{closer} to triggering \(\pi\) (in the sense of affinity), the corresponding \(\AWD\) value strictly decreases, ensuring convergence in the optimization process.

\subsection{AWD Function and its Mathematical Properties}

We define the \(\AWD\) function as any implementation satisfying the interface \(\mathit{Paths} \times \mathit{dom}(\mathit{Prog}) \to \mathbb{F}\) that adheres to the following three properties: non-negativity, zero-target correspondence, and the Monotonic Convergence Condition (MCC). The MCC ensures that for a fixed \(\pi\), if \(\mathit{Aff}(\pi, x)\) decreases, then \(\AWD(\pi, x)\) must also decrease. Since \(\AWD\) returns a scalar value, ensuring compliance with MCC requires an encoding that preserves the lexicographic order \(\prec\).

To achieve this, we construct \(\AWD\) as follows:
\begin{equation}
    \AWD(\pi, x) = u \cdot M + v, \quad \text{where } (u, v) = \mathit{Aff}(\pi, x).
    \label{eq:AWD}
\end{equation}
Here, \(M\) is chosen such that \(M > \max(v)\) for all possible values of \(v\), ensuring correct order preservation. Fig.~\ref{fig:cut} illustrates the fork branch concept and the definition of \(\AWD\).

\begin{wrapfigure}{r}{0.4\textwidth} 
    \centering
    \includegraphics[width=\linewidth]{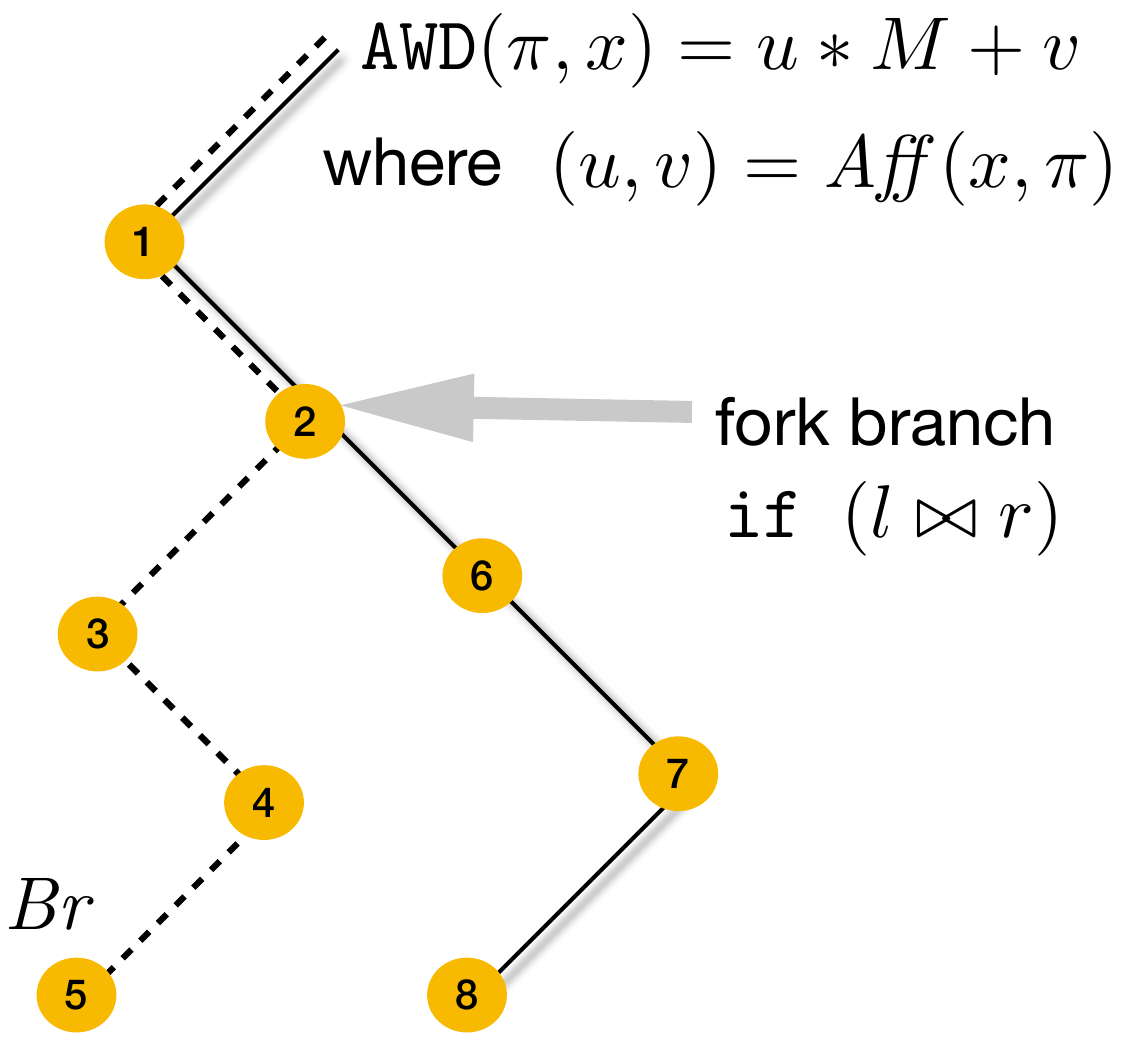}
    \caption{The fork branch and the $\AWD$ function.}
    \label{fig:cut}
\end{wrapfigure}


The function defined in Eq.~\eqref{eq:AWD} satisfies non-negativity and zero-target correspondence since \(u\), \(v\), and \(M\) are non-negative by construction. The crucial property that establishes MCC follows from the lemma below, which ensures that the encoding \(E(u,v)\) correctly preserves the lexicographic order:

\begin{lemma}
Let \((u_1, v_1)\) and \((u_2, v_2)\) be two pairs ordered lexicographically. Define the mapping \(E(u, v) = uM + v\), where \(M > \max(v)\) for all possible \(v\). Then \(E(u, v)\) preserves the lexicographic order:
\begin{enumerate}
    \item If \(u_1 < u_2\), then \(E(u_1, v_1) < E(u_2, v_2)\).
    \item If \(u_1 = u_2\) and \(v_1 < v_2\), then \(E(u_1, v_1) < E(u_2, v_2)\).
\end{enumerate}
\label{lem:lexical}
\end{lemma}
\begin{proof}[Proof Sketch]
We show that the mapping \( E(u, v) = uM + v \), where \( M > \max(v) \), preserves lexicographic order. If \( u_1 < u_2 \), then \( u_1M < u_2M \), and since \( v_1, v_2 \geq 0 \), adding them does not reverse the inequality, ensuring \( E(u_1, v_1) < E(u_2, v_2) \). If \( u_1 = u_2 \) and \( v_1 < v_2 \), then \( E(u_1, v_1) = u_1M + v_1 \) and \( E(u_2, v_2) = u_2M + v_2 \), reducing to \( v_1 < v_2 \), so \( E(u_1, v_1) < E(u_2, v_2) \). The condition \( M > \max(v) \) is necessary; otherwise, a large enough \( v_1 \) could dominate the difference \( (u_2 - u_1)M \), violating order preservation.
\end{proof}

\subsection{Algorithmic Construction of the AWD Function}

To construct the \(\AWD\) function, we need to compute  Eq.~\ref{eq:AWD}. The process of constructing \(\AWD\) follows these steps: (1) We introduce three global variables: \texttt{paths}, \texttt{beats}, and \texttt{d}. The variable \texttt{paths} serves as one of the two inputs to \(\AWD\), preserving the interface \(\mathit{Paths} \times \mathit{dom}(\mathit{Prog}) \to \mathbb{F}\) without altering the original interface of \(\mathit{Prog}\). Similarly, \texttt{d} represents the output of \(\AWD\), ensuring the original output type of \(\mathit{Prog}\) remains intact. The variable \texttt{beats} tracks the execution progress along the expected path, which we use for identifying the fork branch and computing \(u\) in Eq.~\ref{eq:AWD}. (2) Given an input \(x \in \mathit{dom}(\mathit{Prog})\), we require that \(\AWD(\pi, x)\) returns zero if \(x\) follows the expected path \(\pi\); otherwise, it identifies the fork branch, computes \((u, v) = \mathit{Aff}(\pi, x)\), and updates \texttt{d} following Eq.~\ref{eq:AWD}. 

To implement this, we embed a \emph{branch sentinel} within the program under analysis to identify fork branches and compute \((u,v)\) in the path affinity function. The branch sentinel is a callback function that dynamically updates \texttt{d} based on three execution stages: (1) if the realized path matches \(\pi\) completely, then \(\texttt{d} = 0\); (2) if it partially matches, \texttt{d} accumulates a distance proportional to the unmatched branches; (3) if it diverges at a fork, \texttt{d} incorporates the operand distance \(v\), scaled by the remaining unmatched depth. Algorithm~\ref{algo:sen} illustrates the design of the branch sentinel. Embedding the sentinel as a callback function before each branch ensures that the constructed \(\AWD\) function correctly computes Eq.~\ref{eq:AWD} and satisfies the required conditions of non-negativity, zero-target correspondence, and monotonic convergence.

As an implementation detail, we use \( \texttt{v} = \ln (1 + \texttt{v\_original})) \) instead of \texttt{v\_original} itself as one might expect from Def.~\ref{def:path_affinity}. This modification maintains the correctness of the algorithm due to Lem.~\ref{lem:lexical} and ensures that  value of $\texttt{M}$ that is larger than any \texttt{v}, can be found.

\begin{algorithm}[h]

\scriptsize

\caption{Embedded Branch Sentinel: \texttt{sen}}
\KwIn{\texttt{cond}: Branch condition \newline
      \texttt{op}: Comparison operator \newline
      \texttt{lhs}, \texttt{rhs}: Operands for comparison}
\KwOut{\texttt{d}}
\tcp*[l]{\texttt{paths}, \texttt{beats}, \texttt{d} are global variables. \texttt{M} is a large number.}
\If{\texttt{disable\_sen} is true}{
    \Return
}
\texttt{expected\_br} $\gets$ \texttt{paths[beats]} \;
\texttt{realized\_br} $\gets$ \texttt{Compare(op, lhs, rhs)} \;

\texttt{beats} $\gets$ \texttt{beats + 1} \;

\If{\texttt{expected\_br = realized\_br} \textbf{and} \texttt{beats = paths.length}}{
    \texttt{d} $\gets 0$ \tcp*[l]{Entire path matched}
    \texttt{disable\_sen} $\gets$ true \;
}
\ElseIf{\texttt{expected\_br = realized\_br}}{
    \texttt{d} $\gets -42$ \tcp*[l]{Partial progress, placeholder value to be overwritten}
}
\Else{ \tcp*[l]{Fork branch}
    $\texttt{u} \gets \texttt{paths.length} - \texttt{beats}$  \;
    $\texttt{v\_original} \gets \kappa (\texttt{lhs}, \texttt{rhs}, \kappa (\texttt{lhs}, \texttt{rhs}+1, \text{ or 1 depending on } \texttt{op}$  (Def. \ref{def:path_affinity}) \;
    $\texttt{v} \gets \ln (1 + \texttt{v\_original})$ \;
    $\texttt{d} \gets \texttt{u * M + v}$  \;
    \texttt{disable\_sen} $\gets$ true \;
}
\end{algorithm}\label{algo:sen}
\normalsize

%% file: 05_results.tex
\section{Experiments}
\input{04_implem}

\subsection{Experimental Setup}

We evaluated AWD on all 40 benchmarks from SV-COMP 2024's floats-cdfpl category~\cite{svcomp2024,svbenchmarks}, designed for bounds checking~\cite{DBLP:conf/tacas/DSilvaHKT12}. Each C benchmark has a known ground truth (G.T in Tab.~\ref{tab:results}) as either reachable (REA) or unreachable (UNR). 

We compared AWD against two verification tools: static analyzer Astrée~\cite{DBLP:conf/asian/CousotCFMMMR06} and C bounded model checker CBMC~\cite{DBLP:journals/corr/abs-2302-02384} which combines bounded model checking and symbolic execution. Astrée is a proprietary static analysis tool. Since it is not freely available, we did not run it but used recorded results from previous executions on the same set of benchmarks, as documented in~\cite{cprover}. CBMC is an actively developed  tool widely recognized for software verification, particularly for its capability in floating-point verification.

Other potential comparison methods were considered but ultimately excluded. (1) Weak Distance Methods~\cite{DBLP:conf/pldi/FuS19}: While initially considered, weak-distance-based approaches failed to handle even the simplest branching structures, as shown in Sect.~\ref{sect:eg}. (2) Conflict-Driven Learning (CDL)~\cite{DBLP:conf/tacas/DSilvaHKT12}: Originally developed over a decade ago, CDL does not appear to be actively maintained. Additionally, its original authors are now among the primary developers of CBMC. Given CBMC’s status as a leading floating-point verification tool, we consider it the more relevant and representative comparison.

All AWD and CBMC experiments were conducted on a MacBook Air equipped with an Apple M3 chip and 24GB of memory, running macOS Sequoia 15.1. As mentioned above, recorded results from~\cite{cprover} were used for Astrée.

\subsection{Accuracy Comparison}

As shown in Table~\ref{tab:results}, AWD achieved 100\% accuracy, correctly classifying all benchmarks as either reachable (REA) or unreachable (UNR). In contrast, Astrée solved only 17.5\% of the cases, demonstrating its limitations in handling floating-point constraints, while CBMC also achieved 100\% accuracy. These results confirm that AWD provides the same level of precision as model checking while significantly outperforming static analysis.

A sanity check was performed, as shown in columns 3 and 4 of the table. Whenever the AWD function reached a minimum value of \(0\), AWD consistently returned REA, correctly matching the ground truth classification of REA (reachable). This outcome aligns with theoretical expectations and further reinforces the reliability of AWD’s optimization-based approach.

The comparison against Astrée may appear problematic. As a static analysis tool based on over-approximation through abstract interpretation~\cite{DBLP:conf/popl/CousotC79}, Astrée's primary strength lies in proving that erroneous  behaviors will never occur in any execution of the analyzed program. In other words, Astrée is designed to establish unreachability (absence of bugs) rather than verify reachability (presence of bugs), although it can be used for the latter with a tight abstraction~\cite{DBLP:conf/pldi/0002CS21}.  This limitation is evident in the results, as Astrée failed in all cases requiring reachability verification. However, if we consider only the 23 benchmarks where the ground truth is unreachability (UNR), a more appropriate evaluation of Astrée's accuracy emerges. In this case, Astrée correctly identified 7 out of 23, resulting in an adjusted accuracy of 30.4\%, which may be insufficient for practical bounds checking compared to AWD.

 \input{04c_table_all}
 
\subsection{Execution Time Comparison}

AWD demonstrates a substantial performance advantage over CBMC, significantly reducing execution time across all benchmarks. On average, AWD completed each benchmark in 0.55 seconds, whereas CBMC required 94.12 seconds per benchmark. This corresponds to a 170X speedup, highlighting AWD’s efficiency in solving bounds-checking problems.

Astrée was fastest at 0.04s but lacked accuracy, making it unreliable. AWD balances speed and accuracy, providing a practical bounds-checking solution.

CBMC is particularly slow on certain Newton benchmarks; it took 205.25 seconds on \texttt{newton\_2\_7}, while AWD solved it in 0.54 seconds. The most extreme case, \texttt{newton\_3\_3}, required 842.97 seconds for CBMC but only 0.56 seconds for AWD. In contrast, AWD maintains a nearly constant execution time of approximately 0.55 seconds across all benchmarks. This stability likely stems from its execution-based approach, where runtime depends on program execution rather than reasoning overhead. Since all benchmarks are small but vary in complexity, AWD avoids the drastic slowdowns seen in CBMC.

%% file: 04_implem.tex
\section{Implementation}

We implement AWD as a research prototype in Python and C++. The input is a tuple \((M, F, T)\), where \(M\) is an LLVM IR module, \(F\) is the entry function, and \(T\) is the target identifier (e.g., \texttt{\_\_error}). The output is a reachability verdict, optionally providing a triggering input.  

Our implementation consists of three main components. (1) The LLVM pass \texttt{paths.cc} synthesizes execution paths to \(T\) using breadth-first search over the control flow graph, tracking partial paths under depth and call stack constraints. Handling loops and external calls as branches remains a future extension. 
(2) The \texttt{embed.cc} pass instruments conditional branches by inserting type-aware function calls via \texttt{llvm::IRBuilder}, ensuring program semantics remain intact while enabling precise tracking. (3) Global optimization, implemented in \texttt{min.py}, minimizes \(\AWD(\pi, x)\) for each synthesized path, seeking inputs that reach zero. We use \texttt{basinhopping}, a stochastic optimizer leveraging random perturbations and MCMC acceptance criteria, with \texttt{Powell} for local refinement. The AWD function is dynamically loaded as \texttt{euler.so} and invoked via Python’s \texttt{ctypes}, avoiding inter-process overhead. To ensure correctness, we enforce nonnegativity constraints and classify reachability based on the computed minima.

%% file: 04c_table_all.tex


\begin{table}\centering
\ra{1.2}
\setlength{\tabcolsep}{4pt}

\scriptsize
\caption{Benchmark Results: Accuracy and Execution Time of AWD, CBMC, and Astrée for Floating-Point Bounds Checking. 
G. T. refers to the ground truth classification, where UNR indicates an unreachable state and REA indicates a reachable state.}
\label{tab:results}\label{tab:results}

\begin{tabular}{ll|ll|lll|lll} 
\toprule 
\multicolumn{2}{c}{ Benchmark}   & \multicolumn{2}{c}{AWD Results} &   \multicolumn{3}{c}{Accuracy} & \multicolumn{3}{c}{Time (seconds)} \\ 
\cmidrule{1-2} \cmidrule{3-4} \cmidrule{5-7} \cmidrule{8-10} 
Name & G. T. &  Minimum  & Verict  & ASTREE & CBMC & AWD & ASTREE & CBMC & AWD  \\
\midrule 
newton\_1\_1 & UNR & 1.97E-01 & UNR & \greentick & \greentick & \greentick & 0.05 & 8.08 & 0.55 \\
newton\_1\_2 & UNR & 1.77E-01 & UNR & \greentick & \greentick & \greentick & 0.05 & 38.25 & 0.53 \\
newton\_1\_3 & UNR & 1.16E-01 & UNR & \reddash & \greentick & \greentick & 0.05 & 62.73 & 0.53 \\
newton\_1\_4 & REA & 0 & REA & \reddash & \greentick & \greentick & 0.05 & 2.75 & 0.52 \\
newton\_1\_5 & REA & 0 & REA & \reddash & \greentick & \greentick & 0.05 & 1.28 & 0.52 \\
newton\_1\_6 & REA & 0 & REA & \reddash & \greentick & \greentick & 0.05 & 3.97 & 0.54 \\
newton\_1\_7 & REA & 0 & REA & \reddash & \greentick & \greentick & 0.04 & 3.62 & 0.53 \\
newton\_1\_8 & REA & 0 & REA & \reddash & \greentick & \greentick & 0.05 & 3.40 & 0.55 \\
newton\_2\_1 & UNR & 2.00E-01 & UNR & \greentick & \greentick & \greentick & 0.06 & 132.97 & 0.52 \\
newton\_2\_2 & UNR & 2.00E-01 & UNR & \greentick & \greentick & \greentick & 0.06 & 80.84 & 0.55 \\
newton\_2\_3 & UNR & 2.00E-01 & UNR & \reddash & \greentick & \greentick & 0.06 & 78.30 & 0.54 \\
newton\_2\_4 & UNR & 1.96E-01 & UNR & \reddash & \greentick & \greentick & 0.06 & 77.34 & 0.54 \\
newton\_2\_5 & UNR & 1.37E-01 & UNR & \reddash & \greentick & \greentick & 0.06 & 132.82 & 0.56 \\
newton\_2\_6 & REA & 0 & REA & \reddash & \greentick & \greentick & 0.06 & 61.42 & 0.54 \\
newton\_2\_7 & REA & 0 & REA & \reddash & \greentick & \greentick & 0.05 & 205.25 & 0.54 \\
newton\_2\_8 & REA & 0 & REA & \reddash & \greentick & \greentick & 0.04 & 11.47 & 0.57 \\
newton\_3\_1 & UNR & 2.00E-01 & UNR & \greentick & \greentick & \greentick & 0.06 & 182.19 & 0.53 \\
newton\_3\_2 & UNR & 2.00E-01 & UNR & \greentick & \greentick & \greentick & 0.06 & 226.46 & 0.56 \\
newton\_3\_3 & UNR & 2.00E-01 & UNR & \greentick & \greentick & \greentick & 0.06 & 842.97 & 0.56 \\
newton\_3\_4 & UNR & 2.00E-01 & UNR & \reddash & \greentick & \greentick & 0.06 & 223.78 & 0.55 \\
newton\_3\_5 & UNR & 2.00E-01 & UNR & \reddash & \greentick & \greentick & 0.06 & 228.78 & 0.55 \\
newton\_3\_6 & REA & 0 & REA & \reddash & \greentick & \greentick & 0.06 & 105.55 & 0.56 \\
newton\_3\_7 & REA & 0 & REA & \reddash & \greentick & \greentick & 0.05 & 67.35 & 0.54 \\
newton\_3\_8 & REA & 0 & REA & \reddash & \greentick & \greentick & 0.06 & 55.55 & 0.56 \\
sine\_1 & REA & 0 & REA & \reddash & \greentick & \greentick & 0.02 & 0.40 & 0.56 \\
sine\_2 & REA & 0 & REA & \reddash & \greentick & \greentick & 0.02 & 1.38 & 0.53 \\
sine\_3 & REA & 0 & REA & \reddash & \greentick & \greentick & 0.02 & 1.16 & 0.59 \\
sine\_4 & UNR & 1.01E-01 & UNR & \reddash & \greentick & \greentick & 0.02 & 663.80 & 0.54 \\
sine\_5 & UNR & 1.91E-01 & UNR & \reddash & \greentick & \greentick & 0.02 & 15.22 & 0.53 \\
sine\_6 & UNR & 2.91E-01 & UNR & \reddash & \greentick & \greentick & 0.02 & 15.76 & 0.56 \\
sine\_7 & UNR & 5.91E-01 & UNR & \reddash & \greentick & \greentick & 0.02 & 7.22 & 0.55 \\
sine\_8 & UNR & 1.09E+00 & UNR & \reddash & \greentick & \greentick & 0.02 & 16.96 & 0.54 \\
square\_1 & REA & 0 & REA & \reddash & \greentick & \greentick & 0.01 & 1.11 & 0.56 \\
square\_2 & REA & 0 & REA & \reddash & \greentick & \greentick & 0.02 & 0.86 & 0.56 \\
square\_3 & REA & 0 & REA & \reddash & \greentick & \greentick & 0.02 & 0.70 & 0.55 \\
square\_4 & UNR & 1.00E-01 & UNR & \reddash & \greentick & \greentick & 0.02 & 61.41 & 0.55 \\
square\_5 & UNR & 1.00E-01 & UNR & \reddash & \greentick & \greentick & 0.02 & 52.32 & 0.55 \\
square\_6 & UNR & 1.01E-01 & UNR & \reddash & \greentick & \greentick & 0.02 & 52.81 & 0.55 \\
square\_7 & UNR & 1.02E-01 & UNR & \reddash & \greentick & \greentick & 0.02 & 36.01 & 0.57 \\
square\_8 & UNR & 2.02E-01 & UNR & \reddash & \greentick & \greentick & 0.02 & 0.42 & 0.56 \\
\midrule
SUMMARY &               &          &             & 17.50\%  & 100\%      & 100\%      & 0.04 & 94.12 & 0.55 \\
 
\bottomrule 
\end{tabular}

\end{table}

%% file: 06_discussion.tex
\section{Discussion}

\subsubsection*{Theoretical Boundaries and Guarantees}

The augmented weak-distance (AWD) approach extends weak distance but also inherits its theoretical limitations. In general, global optimization may fail to return a true global minimum. To formalize this limitation, we use the notation \( \hat{x}^* \) to denote the global minimum point found by a mathematical optimization backend for a given path \( \pi \), and \( x^* \) to denote the true global minimum. Then, we have:

\[
\AWD(\pi, \hat{x}^*) \geq \AWD(\pi, x^*).
\]

Consider two cases. If \( \AWD(\pi, \hat{x}^*) = \AWD(\pi, x^*) \), meaning the optimization backend finds the correct minimum, then the zero-target correspondence property ensures that AWD produces the correct verdict, whether the path is reachable or unreachable. However, if \( \AWD(\pi, \hat{x}^*) > \AWD(\pi, x^*) \) and additionally \( \AWD(\pi, \allowbreak x^*) = 0 \), then AWD will incorrectly report the path as unreachable, failing to detect a valid bound. Depending on how reachability is framed, this issue may be classified as incompleteness or unsoundness.

\subsubsection*{Loops and External Calls: Design Trade-offs and Challenges}

While AWD mitigates path explosion by ignoring branches within loops and treating external functions as opaque constructs, this simplification risks overlooking target branches that reside within a loop or an external function. For instance, in the \texttt{check\_sum} example (Sect.~\ref{sect:intro}), if the condition \texttt{sum + decimal == 11} were placed inside the loop spanning lines 6 to 97, our current solution would fail to trigger it. This limitation becomes particularly problematic when critical conditions are nested within loops or hidden inside external functions, as these branches are effectively excluded from the path exploration process.

\subsubsection*{Beyond Floating-Point}
Inputs that are not floating-point numbers must be transformed into floating-point representations to leverage the optimization process. This transformation is straightforward for primary data types such as integers or unsigned numbers. However, handling complex types, such as pointers or custom data structures, requires specialized mappings that preserve semantic meaning. Automating this process is non-trivial and inherently limited.

When dealing with conditional branches based on non-floating-point properties, such as \texttt{if Prop(a)}, where \texttt{Prop(a)} evaluates to \texttt{true} or \texttt{false}, the current implementation can only assign a binary distance, which provides no gradient information to guide the optimization process. Otherwise, if \texttt{Prop(a)} is a simple function, it can  be inlined or incorporated directly into the path synthesis process during the path synthesis. For example, this strategy works well for handling branches in the \texttt{safe\_reciprocal} function (Sect.~\ref{sect:eg}). However, when \texttt{Property(a)} involves inaccessible library functions or contains a large number of branches, this approach becomes infeasible.

%% file: 07_relwork.tex

This section reviews related work and positions our contributions. Our approach builds upon the weak distance framework~\cite{DBLP:conf/pldi/FuS19}, which transforms program analysis problems into mathematical optimization tasks by minimizing an objective function. This paradigm has been extensively explored in the programming languages domain, with applications such as stochastic validation of compilers~\cite{DBLP:conf/asplos/Schkufza0A13}, search-based testing~\cite{DBLP:journals/tse/HarmanM10}, floating-point satisfiability solving~\cite{DBLP:conf/cav/FuS16}, and coverage-based testing~\cite{DBLP:conf/pldi/FuS17}. A fundamental distinction between weak distance and traditional heuristic-based search methods lies in theoretical guarantees. While heuristic approaches (e.g., search-based testing) aim to minimize a cost function, they often lack assurances that the computed minima correspond to the intended program behavior~\cite{DBLP:journals/tosem/PereraTAB24}. In contrast, weak distance guarantees that the global minimum corresponds to a valid input triggering the target program behavior. Our work extends this guarantee by introducing a theoretical condition that ensures efficient convergence, thereby improving robustness and usability.

Bounds checking, also known as numerical bound analysis~\cite{DBLP:conf/tacas/DSilvaHKT12}, represents an ideal application domain for AWD due to its emphasis on numerical computations. Unlike symbolic execution~\cite{DBLP:journals/fac/BoerB21}, which encodes floating-point operations using bit-vector logic such as QF\_BVFP~\cite{DBLP:conf/tacas/BrainSS19} and relies on NP-complete decision procedures~\cite{DBLP:series/txtcs/KroeningS16}, AWD turns the challenge of floating-point analysis into an opportunity of leveraging numerical techniques for improved efficiency. While symbolic execution is effective for verifying complex programs involving heap structures~\cite{DBLP:conf/icse/RajputG22}, its generality often leads to scalability challenges, as demonstrated in recent empirical studies~\cite{DBLP:journals/jss/YangZSCW25}. By focusing  on numerical computations, AWD achieves significantly improved performance in  bounds checking.




Traditional tools for verifying numerical properties include static analyzers such as Astrée~\cite{DBLP:conf/asian/CousotCFMMMR06}, model checkers like CBMC~\cite{DBLP:journals/corr/abs-2302-02384}, symbolic execution engines such as KLEE~\cite{DBLP:conf/osdi/CadarDE08}, and SAT/SMT-based approaches like Conflict-Driven Learning (CDL)~\cite{DBLP:conf/tacas/DSilvaHKT12}. Each of these techniques has distinct strengths and limitations. CDL appears to be no longer actively maintained, and weak-distance-based approaches, despite their theoretical appeal, struggle with even simple branching structures, making them impractical for bounds checking. CBMC, as a bounded model checker with a symbolic execution backend, remains one of the most successful tools in floating-point verification and is widely used in verification competitions. Consequently, our evaluation focuses on Astrée and the latest available version of CBMC, as they represent state-of-the-art tools in static analysis, model checking, and symbolic execution for floating-point verification.


More broadly, Mathematical Optimization~\cite{DBLP:books/cu/BV2014,DBLP:conf/osdi/CadarDE08} is a fundamental tool in computer science. It serves as the backbone of modern machine learning, enabling the optimization of objective functions for specific domains, such as loss functions in deep learning~\cite{lecun2015deep} and policy functions in reinforcement learning~\cite{wiering2012reinforcement}. Similarly, weak distance and augmented weak distance focus on designing effective objective functions, but within the context of program analysis. Our work identifies fundamental issues in the weak distance objective function and proposes augmented weak distance to overcome these limitations, ensuring more reliable and efficient optimization in floating-point program verification.

%% file: 08_conc.tex
In this paper, we introduced an augmented weak-distance framework, extending the theoretical foundations of weak-distance optimization to address its key limitations in numerical bounds verification. By overcoming branching discontinuity through a per-path approach,  and overcoming insufficient theoretical conditions by introducing a new condition, Monotonic Convergence Condition (MCC), we ensure practical optimization landscapes that are both theoretically valid and computationally effective while maintaining the theoretical guarantees of the original weak-distance framework.

Our approach was rigorously validated on the SV-COMP 2024 benchmark suite that were used in bounds checking, achieving 100\% accuracy across 40 benchmarks with known ground truths. This significantly outperformed existing techniques, including static analysis, model checking, and conflict-driven learning, both in accuracy and execution time, demonstrating the practical utility of the augmented framework.

The theoretical contributions of our work—particularly the formalization of path-input affinity and the design of the Euler function—lay the groundwork for future extensions. Moving forward, we aim to generalize this framework to handle non-numerical program constructs, such as heap-manipulating programs, while also scaling it to larger, more complex benchmarks.